\documentclass[12pt]{article}
\usepackage{amssymb}
\usepackage{amsmath}
\usepackage{amsthm}
\usepackage{color}
\usepackage[unicode,bookmarks,bookmarksopen,%
bookmarksopenlevel=2,colorlinks,linkcolor=blue,citecolor=green]{hyperref}

\newtheorem{theorem}{Theorem}

\newtheorem{example}{Example}

\newcommand*{\fd}
[2]{\mathchoice{\frac{\delta#1}{\delta#2}}
  {\delta #1/\delta#2}{\delta#1/\delta#2}{\delta#1/\delta#2}}
\newcommand{\ddx}[1]{D^{#1}}
\begin{document}

\title{\textbf{Projective-geometric aspects of homogeneous third-order
    Hamiltonian operators}} \author{E.V. Ferapontov$^1$,
  M.V. Pavlov$^{2, 3}$,
  R.F. Vitolo$^{4}$ \\[3mm]
  $^{1}$Department of Mathematical Sciences,\\
  Loughborough University,\\
  Loughborough, Leicestershire, LE11 3TU, UK\\
  \texttt{e.v.ferapontov@lboro.ac.uk}\\[3mm]
  $^{2}$Department of Mathematical Physics,\\
  Lebedev Physical Institute of Russian Academy of Sciences,\\
  Leninskij Prospekt, 53, Moscow, Russia\\
  \texttt{m.v.pavlov@lboro.ac.uk}\\
  [3mm] $^{3}$Laboratory of Geometrical Methods in Mathematical Physics,\\
Lomonosov Moscow State University,\\
Leninskie Gory 1, 119991 Moscow, Russia\\
 [3mm]
  $^{4}$Department of Mathematics and Physics ``E. De Giorgi'',\\
  University of Salento, Lecce, Italy\\
  \texttt{raffaele.vitolo@unisalento.it} } \date{}
\maketitle

\begin{abstract}
  We investigate homogeneous third-order Hamiltonian operators of
  differential-geometric type. Based on the correspondence with quadratic line
  complexes, a complete list of such operators with
  $n\leq 3$ components is obtained.

  \bigskip


  \bigskip

\end{abstract}

\newpage

\textit{Dedicated to the  memory of Professor Yavuz Nutku (1943-2010)}

\tableofcontents

\section{Introduction}

First-order homogeneous Hamiltonian operators were introduced in \cite{DN} in
the study of one-dimensional systems of hydrodynamic type\footnote{they are
  also known as differential-geometric, or Dubrovin--Novikov brackets}. It was
demonstrated that these operators are parametrized by flat pseudo-Riemannian
metrics.  Higher-order operators were subsequently defined in \cite{DN2}. The
structure of homogeneous second-order Hamiltonian operators was investigated in
\cite{GP87, Doyle}, see also \cite{ferg}.

In this paper we address the problem of classification of homogeneous
third-order Hamiltonian operators of differential-geometric type \cite{GP91,
  Doyle, GP97, OM98, BP},
\begin{multline}
  P=g^{ij}\ddx{3}+b_{k}^{ij}u_{x}^{k}\ddx{2}
  +(c_{k}^{ij}u_{xx}^{k}+c_{km}^{ij}u_{x}^{k}u_{x}^{m})\ddx{}
  \\
  +d_{k}^{ij}u_{xxx}^{k}+d_{km}^{ij}u_{xx}^{k}u_{x}^{m}
  +d_{kmn}^{ij}u_{x}^{k}u_{x}^{m}u_{x}^{n}.
  \label{third}
\end{multline}
Here $u^i$, $i=1$, \dots, $n$, are the dependent (field) variables, and the
coefficients $g^{ij}, \dots, d_{kmn}^{ij}$ depend on $u^{i}$ only; $D$ stands
for the total derivative with respect to $x$. Homogeneity is understood as
follows: the independent variable $x$ has order $-1$, the dependent variables
$u^i$ have order $0$, so that the order of $u^i_x$ and $\ddx{}$ is $1$, etc.
The operator $P$ is Hamiltonian if and only if it is formally skew-adjoint,
$P^*=-P$, and its Schouten bracket vanishes, $[P, P]=0$.  Equivalently, the
corresponding Poisson bracket,
\begin{displaymath}
  \{F_1,F_2\} = \int \fd{F_1}{u^i}P^{ij}\fd{F_2}{u^j} dx,
\end{displaymath}
must be skew-symmetric, and satisfy the Jacobi identity.
We restrict our considerations to the non-degenerate case, $\det g^{ij}\neq
0$. Operators~\eqref{third} are form-invariant under point transformations of
the dependent variables, $u=u(\tilde u)$. Under point transformations, the
coefficients of (\ref{third}) transform as differential-geometric objects.  For
instance, $g^{ij}$ transforms as a $(2,0)$-tensor, so that its inverse $g_{ij}$
defines a pseudo-Riemannian metric (which is not flat in general), the
expressions $ - \frac{1}{3}g_{js}b_{k}^{si}$, $ - \frac{1}{3}g_{js}c_{k}^{si}$,
$ - g_{js}d_{k}^{si}$ transform as Christoffel symbols of affine connections,
etc.
It was conjectured in \cite{Novikov} that the last connection, $\Gamma
_{jk}^{i}= - g_{js}d_{k}^{si}$, must be symmetric and flat; this was confirmed
in \cite{GP91}, see also \cite{Doyle}.  Therefore, there exists a coordinate
system (flat coordinates) such that $\Gamma_{jk}^{i}$ vanish. These coordinates
are determined up to affine transformations. We will keep for them the same
notation $u^i$, note that $u^i$ are nothing but the densities of Casimirs of
the corresponding Hamiltonian operator (\ref{third}). In the flat coordinates
the last three terms in (\ref{third}) vanish, leading to the simplified
expression
\begin{equation}
  P=\ddx{}\left(g^{ij}\ddx{}+c_{k}^{ij}u_{x}^{k}\right)\ddx{}.  \label{casimir}
\end{equation}
This operator is Hamiltonian if and only if the coefficients $g^{ij}$ and
$c_{k}^{ij}$ satisfy the following relations:
\begin{subequations}
  \begin{gather}
    g_{,k}^{ij}=c_{k}^{ij}+c_{k}^{ji},\label{eq:1}\\
    c_{s}^{ij}g^{sk}=-c_{s}^{kj}g^{si},  \label{a}\\
    c_{s}^{ij}g^{sk}+c_{s}^{jk}g^{si}+c_{s}^{ki}g^{sj}=0,\label{eq:2}\\
    c_{s,m}^{ij}g^{sk}=c_{s}^{ik}c_{m}^{sj}-c_{s}^{ki}c_{m}^{sj}-c_{s}^{kj}g_{,m}^{si}.
    \label{b}
  \end{gather}
\end{subequations}
Here \eqref{eq:1} is equivalent to $P^*=-P$, while \eqref{a}-\eqref{b} are
equivalent to $[P, P]=0$.  These conditions are invariant under affine
transformations of the flat coordinates.
It is useful to rewrite the above system in low indices. Introducing
$c_{ijk}=g_{iq}g_{jp}c_{k}^{pq}$ one obtains \cite{GP97}:
\begin{subequations}\label{eq:3}
  \begin{gather}
    g_{mn,k}=-c_{mnk}-c_{nmk},\label{eq:4}\\
    c_{mnk}=-c_{mkn},\label{eq:5}
    \\
    c_{mnk}+c_{nkm}+c_{kmn}=0,\label{eq:6}\\
    c_{mnk,l}=-g^{pq}c_{pml}c_{qnk}.  \label{eq:7}
  \end{gather}
\end{subequations}
Our main observation is that the metric $g$ satisfying equations \eqref{eq:3}
must be the Monge metric of a quadratic line complex.  Since complexes of lines
belong to projective geometry, equations \eqref{eq:3} should
be invariant under the full projective (rather than affine) group. We
demonstrate that this is indeed the case. Based on the projective
classification of quadratic line complexes in $\mathbb{P}^3$ into eleven Segre
types \cite{Jess}, we give a complete list of three-component Hamiltonian
operators.


The structure of the paper is as follows. After discussing known examples of
third-order homogeneous Hamiltonian operators in Section~\ref{sec:examples}, we
sumarize our main results in Section~\ref{sec:summary-main-results}. In
Section~\ref{sec:geom-interpr-poiss} we establish a link between homogeneous
third-order Hamiltonian operators and Monge metrics/quadratic
line complexes. This indicates that the theory is essentially
projectively-invariant (Section \ref{sec:proj}), and leads to the
classification results presented in Section~\ref{sec:class-results}.

All computations were performed with the software package CDIFF \cite{cdiff} of
the REDUCE computer algebra system \cite{reduce}.

\section{Examples}\label{sec:examples}

To the best of our knowledge, all interesting examples of integrable systems
possessing Hamiltonian structures of the form (\ref{third}) come from the
theory of Witten--Dijkgraaf--Verlinde--Verlinde (WDVV) equations of 2D
topological field theory. These are integrable PDEs of Monge--Amp\`ere type
which acquire a Hamiltonian formulation upon transformation into hydrodynamic
form \cite{FM}.


\begin{example} \cite{OM98} The hyperbolic Monge--Amp\`ere equation,
  $u_{tt}u_{xx}-u_{xt}^2= - 1$, can be reduced to hydrodynamic form
  \begin{equation*}
    a_{t}=b_{x},\text{ \ }b_{t}=\left( \frac{b^{2}-1}{a}\right) _{x},
  \end{equation*}
  via the change of variables $a=u_{xx}$, $b=u_{xt}$. It possesses the
  Hamiltonian formulation
  \begin{equation*}
    \begin{pmatrix}
      a \\ b
    \end{pmatrix}_t = P\begin{pmatrix} \fd{H}{a}\\ \fd{H}{b}
    \end{pmatrix},
  \end{equation*}
  with the homogeneous third-order Hamiltonian operator
  \begin{equation*}
    P=  \ddx{}\left(
      \begin{array}{cc}
	0 & \displaystyle\ddx{}\frac{1}{a} \\
	\displaystyle\frac{1}{a}\ddx{} & \displaystyle
	\frac{b}{a^{2}}\ddx{}+\ddx{}\frac{b}{a^{2}}
      \end{array}%
    \right) \ddx{},
  \end{equation*}
  and the nonlocal Hamiltonian,
  \begin{equation*}
    H= - \int \left( \frac{1}{2}a(\ddx{-1}b)^2+\ddx{-2}a\right) dx.
  \end{equation*}
  Note that $\fd{H}{a}=-\frac{1}{2}(\ddx{-1}b)^2-\frac{x^2 }{2}, \
  \fd{H}{b}=\ddx{-1}(a\ddx{-1}b).$
\end{example}

\begin{example}\label{sec:examples-1}
  \cite{FN} The simplest nontrivial case of the WDVV equations is the
  third-order Monge--Amp\`ere equation, $f_{ttt} = f_{xxt}^2 - f_{xxx}f_{xtt}$
  \cite{Dub}. This PDE can be transformed into hydrodynamic form,
  \begin{equation*}
    a_t=b_x,\quad b_t=c_x,\quad c_t=(b^2-ac)_x,
  \end{equation*}
  via the change of variables $a=f_{xxx}$, $b=f_{xxt}$, $c=f_{xtt}$. This
  system possesses the Hamiltonian formulation
  \begin{equation*}
    \begin{pmatrix}
      a \\ b \\ c
    \end{pmatrix}_t =P
    \begin{pmatrix}
      \fd{H}{a}\\ \fd{H}{b} \\ \fd{H}{c}
    \end{pmatrix},
  \end{equation*}
  with the homogeneous third-order Hamiltonian operator
  $$
  P= \ddx{}\left(
    \begin{array}{ccc}
      0 & 0 & \displaystyle\ddx{} \\
      0 & \displaystyle\ddx{} & -\displaystyle\ddx{}a \\
      \displaystyle\ddx{} & -a\displaystyle\ddx{} &
      \displaystyle\ddx{}b + b\ddx{} + a \ddx{} a
    \end{array}\right)
  \ddx{},
    $$
    and the nonlocal Hamiltonian,
    \begin{displaymath}
      H=-\int\left(  \frac{1}{2}a\left({\ddx{}}^{-1}b\right)^2 + {\ddx{}}^{-1}b
	{\ddx{}}^{-1}c\right)dx.
    \end{displaymath}
  \end{example}

\begin{example}
  Further examples are provided in \cite{KN1,KN2}.  One of them is the equation
  $f_{xxx} = f_{ttx}^2 - f_{ttt}f_{txx}$ which is obtained from the WDVV
  equation of Example 2 by simply interchanging $t$ and $x$. Remarkably, the
  corresponding Hamiltonian formulation is rather different. The change of
  variables $a=f_{xxx}$, $b=f_{xxt}$, $c=f_{xtt}$ brings the equation into
  hydrodynamic form,
  \begin{equation*}
    a_t=b_x,\quad b_t=c_x,\quad c_t=\left(\frac{c^2-a}{b}\right)_x.
  \end{equation*}
  This system possesses the Hamiltonian formulation
  \begin{equation*}
    \begin{pmatrix}
      a \\ b \\ c
    \end{pmatrix}_t =P
    \begin{pmatrix}
      \fd{H}{a}\\ \fd{H}{b} \\ \fd{H}{c}
    \end{pmatrix},
  \end{equation*}
  with the homogeneous third-order Hamiltonian operator
$$
P= \ddx{}\left(
  \begin{array}{ccc}
    -\displaystyle\ddx{} & 0 & 0 \\[3mm]
    0 & 0 & \displaystyle\ddx{}\frac{1}{b} \\[3mm]
    0 & \displaystyle \frac{1}{b}\ddx{} &
    \displaystyle \frac{c}{b^2}\ddx{} +
    \ddx{}\frac{c}{b^2}
  \end{array}\right)
\ddx{},
 $$
 and the nonlocal Hamiltonian,
 \begin{displaymath}
   H=\int\left( c {\ddx{}}^{-1} b {\ddx{}}^{-1} c +
     {\ddx{}}^{-1}a{\ddx{}}^{-1}b\right)dx.
 \end{displaymath}
 Note that this operator is a direct sum of the one-component operator
 $-\ddx{3}$ (in the variable $a$), and the two-component operator from Example
 1 (in variables $b, c$). On the contrary, the operator from Example 2 is not
 reducible.
\end{example}


Two more WDVV-type equations were considered in \cite{KN1,KN2}, namely
$f_{ttt}+f_{ttt}f_{xxx}-f_{ttx}f_{txx}+f_{ttt}f_{txx}-f_{xtt}^2+f_{xxx}f_{xtt}
-f_{xxt}^2=0$, and the corresponding equation obtained by exchanging $t$ and
$x$. Both equations admit homogeneous third-order Hamiltonian structures which
are equivalent to the one from Example 3. Operators from Examples 1-3 will
feature in the classification results below.

\section{Summary of main results}
\label{sec:summary-main-results}

Our first observation (Proposition 1 of Section \ref{sec:geom-interpr-poiss})
is that equations \eqref{eq:3} can be rewritten in terms of the metric $g$
alone, implying the linear subsystem
\begin{equation}
  g_{mk,n}+g_{kn,m}+g_{mn,k}=0,
  \label{Killing}
\end{equation}
along with a more complicated set of nonlinear constraints,
\begin{equation}
  \begin{array}{c}
    g_{m[k,n]l}=-\frac{1}{3}g^{pq} g_{p[l,m]}g_{q[k,n]},
  \end{array}
  \label{nonlin}
\end{equation}
where square brackets denote antisymmetrisation. Any solution to these
equations specifies a third-order Hamiltonian operator of the form
(\ref{casimir}) by setting $ c_{nkm}=\frac{1}{3}g_{n[m,k]}$.

Our second remark is that the generic metric $g=g_{ij}du^idu^j$ satisfying the
linear subsystem (\ref{Killing}) is an arbitrary quadratic expression in $du^i$
and $u^jdu^k-u^kdu^j$, explicitly,
\begin{equation}
  \begin{array}{c}
    g_{ij}du^idu^j=a_{ij}du^idu^j + b_{ijk}du^i(u^jdu^k-u^kdu^j) +
    \\
    c_{ijkl}(u^idu^j-u^jdu^i)(u^kdu^l-u^ldu^k),
  \end{array}
  \label{Monge}
\end{equation}
where $a_{ij}, \ b_{ijk}, \ c_{ijkl}$
are arbitrary constants.

Since the flat coordinates are defined up to affine transformations, the system
(\ref{Killing})-(\ref{nonlin}) is invariant under point transformations of the
form
$$
\tilde u^i= l^i({\bf u}), ~~~ \tilde g= g,
$$
where $l^i$ are arbitrary linear forms in the flat coordinates ${\bf u}=(u^1,
\dots, u^n)$, and $\tilde g=g$ indicates that $g$ transforms as a metric. What
is less obvious is that the system (\ref{Killing})-(\ref{nonlin}) is invariant
under the bigger group of projective transformations,
$$
\tilde u^i= \frac{l^i({\bf u})}{l({\bf u})}, ~~~ \tilde g= \frac{g}{l^4({\bf
    u})},
$$
where $l$ is yet another linear form in the flat coordinates.  It will be
demonstrated in Section \ref{sec:proj} that projective transformations
correspond to reciprocal transformations of the Hamiltonian operator
(\ref{casimir}).  Note that the ansatz (\ref{Monge}) is invariant under
projective transformations indicated above.  One can thus formulate two natural
classification problems: affine and projective classifications.

Metrics of the form (\ref{Monge}) typically arise as Monge metrics of quadratic
line complexes. Recall that a quadratic line complex is a $(2n-3)$-parameter
family of lines in the projective space $\mathbb{P}^n$ specified by a single
quadratic equation in the Pl\"ucker coordinates. Fixing a point $p\in
\mathbb{P}^n$ and taking all lines of the complex which pass though $p$ we
obtain a quadratic cone with vertex at $p$. This field of cones supplies
$\mathbb{P}^n$ with a conformal structure (Monge metric) whose general form is
given by (\ref{Monge}), see Section \ref{sec:geom-interpr-poiss} for more
details. The key invariant of a quadratic line complex is its singular variety
(which is a hypersurface in $\mathbb{P}^n$ of degree $2n-2$, see
\cite{Dolgachev}, Prop. 10.3.2), defined by the equation
$$
\det g_{ij}=0.
$$
For $n=2$ the singular variety is a conic in $\mathbb{P}^2$, for $n=3$ it is
the Kummer quartic in $\mathbb{P}^3$, etc.

Taking a generic Monge metric (\ref{Monge}), bringing it to a suitable normal
form via affine/projective transformations, and verifying the remaining
nonlinear constraints (\ref{nonlin}) one can obtain a classification of
third-order Hamiltonian operators.  Due to the complexity of nonlinear
constraints, we only managed to complete this programme in two- and
three-component cases (Section \ref{sec:class-results}), note that any
one-component operator is equivalent to $\ddx{3}$. We observe that the singular
varieties of Monge metrics corresponding to homogeneous third-order Hamiltonian
operators degenerate into double hypersurfaces of degree $n-1$.  Our
classification results are summarised below (in the two-component case we give
both affine and projective classifications, in the three-component situation
the affine classification contains too many special cases and moduli, and is
omitted):

\medskip

\noindent {\bf Two-component case} (Theorem \ref{sec:two-component-case-1} of
Section \ref{sec:class-results}). {\it Modulo (complex) affine transformations,
  the metric of any two-component homogeneous third-order Hamiltonian operator
  can be reduced to one of the three canonical forms}:
\begin{displaymath}
  g^{(1)}=\left(
    \begin{array}{cc}
      (u^{2})^{2}+1 &-u^{1}u^{2} \\-u^{1}u^{2} & (u^{1})^{2}
    \end{array}
  \right), ~~
  g^{(2)}=\left(
    \begin{array}{cc}
      -2u^{2} & u^{1} \\
      u^{1} & 0
    \end{array}
  \right), ~~
  g^{(3)}=\left(
    \begin{array}{cc}
      1 & 0 \\
      0 & 1
    \end{array}
  \right).
\end{displaymath}
The metric $g^{(2)}$ corresponds to the third-order Hamiltonian operator from
Example 1 of Section \ref{sec:examples}. One can verify that the metric
$g^{(2)}$ is flat, while $g^{(1)}$ is not flat. The singular varieties of the
first two metrics are double lines: $(u^1)^2=0$. Applying a
projective transformation which sends this line to the line at infinity, one
can reduce the first two cases to constant coefficients. This leads to our
second result (Theorem \ref{sec:two-component-case-2}):

{\it Modulo projective transformations, any two-component homogeneous
  third-order Hamiltonian operator can be reduced to constant form.  }

\medskip

\noindent {\bf Three-component case} (Theorem \ref{sec:three-component-case-1}
of Section \ref{sec:class-results}).  {\it Modulo (complex) projective
  transformations, the metric of any three-component homogeneous third-order
  Hamiltonian operator can be reduced to one of the six canonical forms}:
\begin{gather*}\footnotesize
  g^{(1)}=\begin{pmatrix} (u^{2})^{2}+c & -u^{1}u^{2}-u^{3} & 2u^{2} \\
    -u^{1}u^{2}-u^{3} & (u^{1})^{2}+c(u^{3})^{2} & -cu^{2}u^{3}-u^{1} \\ 2u^{2}
    & -cu^{2}u^{3}-u^{1} & c(u^{2})^{2}+1
  \end{pmatrix}, \\ \footnotesize g^{(2)} = \begin{pmatrix}
    (u^{2})^{2}+1 & -u^{1}u^{2}-u^{3} & 2u^{2} \\
    -u^{1}u^{2}-u^{3} & (u^{1})^{2} & -u^{1} \\
    2u^{2} & -u^{1} & 1
  \end{pmatrix}, \quad g^{(3)} = \begin{pmatrix}
    (u^{2})^{2}+1 &  -u^{1}u^{2}&0 \\
    -u^1u^2 & (u^1)^2 & 0 \\
    0 & 0 & 1%
  \end{pmatrix}, \\ \footnotesize g^{(4)}= \begin{pmatrix} -2u^2 & u^1 & 0
    \\
    u^1 & 0 & 0
    \\
    0 & 0 & 1
  \end{pmatrix}, \quad g^{(5)}=\begin{pmatrix} -2u^2 & u^1 & 1
    \\
    u^1 & 1 &0
    \\
    1 & 0 & 0
  \end{pmatrix}, \quad g^{(6)} =
  \begin{pmatrix}
    1 & 0 & 0\\ 0 & 1 & 0\\ 0 & 0 & 1
  \end{pmatrix}.
\end{gather*}
The corresponding singular varieties, $\det g=0$, are as follows (see Theorem
\ref{sec:three-component-case-1} for explicit formulae):
\begin{itemize}
\item $g^{(1)}, g^{(2)}$: double quadric;
\item $g^{(3)}, g^{(4)}$: two double planes, one of them at infinity;
\item $g^{(5)}, g^{(6)}$: quadruple plane at infinity.
\end{itemize}
Third-order Hamiltonian operators corresponding to the metrics $g^{(3)}$ and
$g^{(4)}$ are direct sums of the two-component operators from Theorem
\ref{sec:two-component-case-1}, and the one-component operator $\ddx{3}$ (we
emphasize that these direct sums can't be transformed to constant form, even
by projective transformations). As the correspondence between  Monge
metrics and  Hamiltonian operators \eqref{casimir} respects direct sums,  two-component operators are expected to appear in the three-component
classification. The metrics $g^{(5)}$ and $g^{(4)}$ correspond to Hamiltonian
operators discussed in Examples 2, 3 of Section \ref{sec:examples}. The metrics
$g^{(1)}$ and $g^{(2)}$ give rise to third-order operators which are apparently
new. Direct calculations demonstrate that the metrics $g^{(4)}, g^{(5)},
g^{(6)}$ are flat, while $g^{(1)}, g^{(2)}, g^{(3)}$ are not flat (not even
conformally flat: they have non-vanishing Cotton tensor).


\section{Monge metrics and quadratic line complexes}
\label{sec:geom-interpr-poiss}

Our first observation is that system \eqref{eq:3} can be rewritten in terms of
the metric $g$ alone:


\medskip

\noindent {\bf Proposition 1}.
{\it The system \eqref{eq:3} implies
$$
c_{nkm}=\frac{1}{3}(g_{nm,k}-g_{nk,m})=\frac{1}{3}g_{n[m,k]},
$$
and the elimination of $c$ results in (\ref{Killing}), (\ref{nonlin}):
$$
g_{mk,n}+g_{kn,m}+g_{mn,k}=0,\label{eq:18}
 $$
 $$
 g_{m[k,n]l}=-\frac{1}{3}g^{pq} g_{p[l,m]}g_{q[k,n]},
  $$
  here square brackets denote antisymmetrisation. }

\begin{proof}

  Taking into account that $c$ is skew-symmetric in the last two indices, the
  relation \eqref{eq:4} implies
 $$
 c_{mnk}=c_{nkm}-g_{mn, k}, ~~~ c_{kmn}=c_{nkm}+g_{kn, m}.
 $$
 Substituting this into \eqref{eq:6} we obtain the explicit formula for $c$,
 $$
 c_{nkm}=\frac{1}{3}(g_{nm,k}-g_{nk,m})=\frac{1}{3}g_{n[m,k]}.
  $$
  With this expression for $c$, the relations \eqref{eq:4}-\eqref{eq:6} reduce
  to the linear system (\ref{Killing}) for $g$:
$$
g_{mk,n}+g_{kn,m}+g_{mn,k}=0.
$$
Finally, \eqref{eq:7} gives the nonlinear constraint (\ref{nonlin}).
\end{proof}

Note that the linear system (\ref{Killing}) can be solved explicitly: any such
metric $g=g_{ij}du^idu^j$ is an arbitrary quadratic expression of the form
(\ref{Monge}) in $du^i$ and $u^jdu^k-u^kdu^j$:
$$
\begin{array}{c}
  g_{ij}du^idu^j=a_{ij}du^idu^j + b_{ijk}du^i(u^jdu^k-u^kdu^j) +
  \\
  c_{ijkl}(u^idu^j-u^jdu^i)(u^kdu^l-u^ldu^k),
\end{array}
  $$
  here the coefficients $a_{ij}, b_{ijk}, c_{ijkl}$ are arbitrary constants
  (without any loss of generality one can impose additional symmetries such as
  $a_{ij}=a_{ij}$, $b_{ijk}=-b_{ikj}$, $c_{ijkl} = -c_{jikl} = -c_{ijlk}$ ,
  etc).  The above formula follows from the analogous result for Killing
  bivectors in pseudo-Euclidean spaces: any Killing bivector is a quadratic
  expression in Killing vectors.  Formula (\ref{Monge}) implies that the
  coefficients $g_{ij}$ are at most quadratic in the flat coordinates $u^i$,
  the fact observed previously in \cite{GP91, Doyle}.

  Metrics of the form (\ref{Monge}) appear in the theory of quadratic complexes
  of lines in the projective space $\mathbb{P}^n$. Let us recall the main
  construction. Consider two points in $\mathbb{P}^n$ with homogeneous
  coordinates $u^i, v^i$, $i=1, \dots, n+1$. The Pl\"ucker coordinates $p^{ij}$
  of the line through these points are defined as $p^{ij}=u^iv^j - u^jv^i$.
  They satisfy a system of quadratic relations of the form
  $p^{ij}p^{kl}+p^{ki}p^{jl}+p^{jk}p^{il}=0$, which specify a projective
  embedding of the Grassmannian of lines (Pl\"ucker embedding). For $n=3$ we
  have a single quadratic relation $ p^{12}p^{34} + p^{31}p^{24} +
  p^{14}p^{23}=0$, known as the Pl\"ucker quadric.  A quadratic line complex is
  defined by an additional homogeneous quadratic equation in the Pl\"ucker
  coordinates,
$$
Q(p^{ij})=0.
$$
This specifies a $(2n-3)$-parameter family of lines in $\mathbb{P}^n$. Fixing a
point $p\in \mathbb{P}^n$ and taking all lines of the complex which pass though
$p$ we obtain a quadratic cone with vertex at $p$. This family of cones
supplies $\mathbb{P}^n$ with a conformal structure (Monge metric) whose
explicit form can be obtained as follows. Let us set $v^i=u^i+du^i$.  Then the
Pl\"ucker coordinates take the form $p^{ij}=u^idu^j-u^jdu^i$. In the affine
chart $u^{n+1}=1, \ du^{n+1}=0$, part of the Pl\"ucker coordinates simplify to
$p^{(n+1) i}=du^i$, and the equation of the complex takes the so-called Monge
form:
$$
Q(du^i, u^jdu^k-u^kdu^j)=0,
$$
here $i, j, k=1, \dots, n$. This is nothing but the general metric
(\ref{Monge}).
What renders the classification of three-component homogeneous third-order
Hamiltonian operators possible, is the existing classification of quadratic
line complexes in $\mathbb{P}^3$ \cite{Jess}.

\section{Projective invariance and reciprocal transformations }
\label{sec:proj}

As the flat coordinates $u^i$ are defined up to affine transformations, the
system (\ref{Killing})-(\ref{nonlin}) is invariant under transformations of the
form
$$
\tilde u^i= l^i({\bf u}), ~~~ \tilde g= g,
$$
where $l^i$ are linear forms in the flat coordinates ${\bf u}=(u^1, \dots,
u^n)$, and $\tilde g=g$ indicates that $g$ transforms as a metric (with low
indices). On the other hand, the relation to quadratic line complexes indicates
that our problem is projectively-invariant. Indeed, the system
(\ref{Killing})-(\ref{nonlin}) is invariant under the group of projective
transformations of the form
\begin{equation}
  \tilde u^i= \frac{l^i({\bf u})}{l({\bf u})}, ~~~ \tilde g= \frac{g}{l^4({\bf u})},
  \label{pr}
\end{equation}
where $l$ is yet another linear form in the flat coordinates. Note that the
Monge form (\ref{Monge}) is also invariant under projective transformations
(\ref{pr}).

It turns out that projective transformations (\ref{pr}) correspond to
reciprocal transformations of the corresponding Hamiltonian operator
(\ref{casimir}).  We recall that a reciprocal transformation is a nonlocal
change of the independent variable $x$ defined as
\begin{equation}
  d\tilde x=A({\bf u}) dx,
  \label{recip}
\end{equation}
where $A({\bf u})$ is a function of field variables. Reciprocal transformations
of Hamiltonian operators of hydrodynamic type were investigated previously in
\cite{Fer95, fp, Abenda}. In general, transformed operators become nonlocal. It
is remarkable that in the special case when $A({\bf u})$ is linear in the flat
coordinates, reciprocal transformations preserve the locality of third-order
operators (\ref{casimir}).

\medskip

\noindent{\bf Proposition 2.} {\it
  The class of homogeneous third-order  Hamiltonian operators \eqref{casimir} is
  invariant under reciprocal transformations of the form (\ref{recip}), where
  $A({\bf u})$ is linear in the flat coordinates $u^i$. Reciprocal
  transformations induce projective transformations (\ref{pr}) of the corresponding Monge
  metrics.}

\begin{proof}
Let us set $A=c_iu^i+c_0$. In the new independent variable $\tilde x$, the Casimir functionals, $\int u^i dx$, take the form $\int\frac{u^i}{A} d\tilde x$.
Thus, the transformed Casimir densities are $\tilde u^i=\frac{u^i}{A}$, which is a particular case of  (\ref{pr}). The general case of (\ref{pr}) is obtained by combining the above transformation  with arbitrary affine changes of $u^i$.

The second formula (\ref{pr}) results from the following calculation. Consider two functionals,
$F=\int f({\bf u}) dx$ and $H=\int h({\bf u}) dx$ (for simplicity we restrict to functionals of hydrodynamic type). Their Poisson bracket equals
$$
\{ F, H\}=\int f_iP^{ij}h_j dx=\int f_i\ddx{}(g^{ij}\ddx{}+c^{ij}_ku^k_x)\ddx{} h_j dx.
$$
Using $\int f dx=\int \tilde f d\tilde x, \ \int h dx=\int \tilde h d\tilde x$ where $f=A\tilde f, \ h=A\tilde h$, and making the substitutions  $dx\to \frac{1}{A} d\tilde x, \  \ddx{} \to A\tilde {\ddx{}}$ (here $\tilde {\ddx{}}\equiv \ddx{}_{\tilde x}$), one obtains
$$
\{ F, H\}=\int (A\tilde f_i+c_i \tilde f){\underline A}\tilde {\ddx{}}(Ag^{ij}\tilde {\ddx{}}+Ac^{ij}_ku^k_{\tilde x})A\tilde {\ddx{}} (A\tilde h_j+c_j\tilde h) \underline{\frac{1}{A}}d\tilde x.
$$
Cancelling the underlined terms one can observe that, in spite of the explicit presence of $\tilde f$ and $\tilde h$ in the integrand (which may potentially lead to non-locality of the transformed operator), the fact that they appear with constant coefficients $c_i$ allows one to rewrite the above expression in the form
$$
\{ F, H\}=\int \tilde f_i\tilde P^{ij}\tilde h_j d\tilde x,
$$
where $\tilde P$ is a local homogeneous third-order operator with the leading term $A^4g^{ij}\tilde{\ddx{3}}$.
Thus, $\tilde g^{ij}=A^4g^{ij}$, which is equivalent to the second formula (\ref{pr}) (recall that in (\ref{pr}) $g$ denotes the metric with low indices).

In the new Casimirs $\tilde u^i$, the transformed operator $\tilde P$ can be computed in the following way. Taking into account that
$
\tilde{u}^{k}=\frac{u^{k}}{A},
$
where $A=c_{m}u^{m}+c_{0}=c_{0}(1-c_{m}\tilde{u}^{m})^{-1}$, one obtains%
\begin{equation*}
\frac{\partial \tilde{u}^{k}}{\partial u^{i}}=\frac{\delta _{i}^{k}-c_{i}%
\tilde{u}^{k}}{A}.
\end{equation*}%
Thus,
\begin{equation*}
A\frac{\partial \tilde{f}}{\partial u^{i}}=(\delta _{i}^{k}-c_{i}\tilde{u}%
^{k})\frac{\partial \tilde{f}}{\partial \tilde{u}^{k}},
\end{equation*}%
so that the bracket
\begin{equation*}
\{F,H\}=\int \left( A\frac{\partial \tilde{f}}{\partial u^{i}}+c_{i}\tilde{f}%
\right) \tilde{D}(Ag^{ij}\tilde{D}+Ac_{k}^{ij}u_{\tilde{x}}^{k})A\tilde{D}%
\left( A\frac{\partial \tilde{h}}{\partial u^{j}}+c_{j}\tilde{h}\right) d%
\tilde{x}
\end{equation*}%
assumes the form
\begin{equation*}
\{F,H\}=\int \left( (\delta _{i}^{k}-c_{i}\tilde{u}^{k})\frac{\partial
\tilde{f}}{\partial \tilde{u}^{k}}+c_{i}\tilde{f}\right) \tilde{D}(Ag^{ij}%
\tilde{D}+Ac_{k}^{ij}u_{\tilde{x}}^{k})A\tilde{D}\left( (\delta
_{j}^{k}-c_{j}\tilde{u}^{k})\frac{\partial \tilde{h}}{\partial \tilde{u}^{k}}%
+c_{j}\tilde{h}\right) d\tilde{x}.
\end{equation*}%
Using the identity
\begin{equation*}
\tilde{D}\left( (\delta _{j}^{k}-c_{j}\tilde{u}^{k})\frac{\partial \tilde{h}%
}{\partial \tilde{u}^{k}}+c_{j}\tilde{h}\right) =(\delta _{j}^{k}-c_{j}%
\tilde{u}^{k})\tilde{D}\frac{\partial \tilde{h}}{\partial \tilde{u}^{k}},
\end{equation*}%
one obtains%
\begin{equation*}
\{F,H\}=\int \frac{\partial \tilde{f}}{\partial \tilde{u}^{k}}(\delta
_{i}^{k}-c_{i}\tilde{u}^{k})\tilde{D}(Ag^{ij}\tilde{D}+Ac_{k}^{ij}u_{\tilde{x%
}}^{k})A(\delta _{j}^{k}-c_{j}\tilde{u}^{k})\tilde{D}\frac{\partial \tilde{h}%
}{\partial \tilde{u}^{k}}d\tilde{x}
\end{equation*}%
\begin{equation*}
-\int c_{i}\tilde{D}\tilde{f}\cdot (Ag^{ij}\tilde{D}+Ac_{k}^{ij}u_{\tilde{x}%
}^{k})A(\delta _{j}^{k}-c_{j}\tilde{u}^{k})\tilde{D}\frac{\partial \tilde{h}%
}{\partial \tilde{u}^{k}}d\tilde{x}.
\end{equation*}%
Since $\tilde{D}\tilde{f}=(\partial \tilde{f}/\partial \tilde{u}^{m})\tilde{u%
}_{\tilde{x}}^{m}$, one arrives at%
\begin{equation*}
\{F,H\}=\int \frac{\partial \tilde{f}}{\partial \tilde{u}^{m}}[(\delta
_{i}^{m}-c_{i}\tilde{u}^{m})\tilde{D}-c_{i}\tilde{u}_{\tilde{x}}^{m}](Ag^{ij}%
\tilde{D}+Ac_{k}^{ij}u_{\tilde{x}}^{k})A(\delta _{j}^{k}-c_{j}\tilde{u}^{k})%
\tilde{D}\frac{\partial \tilde{h}}{\partial \tilde{u}^{k}}d\tilde{x}.
\end{equation*}%
Using the identity
\begin{equation*}
(\delta _{i}^{m}-c_{i}\tilde{u}^{m})\tilde{D}-c_{i}\tilde{u}_{\tilde{x}}^{m}=%
\tilde{D}(\delta _{i}^{m}-c_{i}\tilde{u}^{m}),
\end{equation*}%
one ultimately arrives at the  Hamiltonian operator $\tilde P$  written in the transformed Casimirs $\tilde u^i$:%
\begin{equation*}
\tilde{P}^{ij}=\tilde{D}(\delta _{i}^{m}-c_{i}\tilde{u}^{m})A(g^{ij}\tilde{D}%
+c_{k}^{ij}u_{\tilde{x}}^{k})A(\delta _{j}^{k}-c_{j}\tilde{u}^{k})\tilde{D}.
\end{equation*}%
This is again a local homogeneous  third-order expression of the form (\ref{casimir}).
\end{proof}

\section{Classification results}
\label{sec:class-results}

In this section we classify homogeneous third-order Hamiltonian operators with
the number of components $n=1,2$ and $3$.  Our approach is based on the
correspondence with Monge metrics and quadratic line complexes. All
classification results are obtained modulo (complex) projective transformations
as introduced in Section \ref{sec:proj}.  To save space we only present
canonical forms for the corresponding Monge metrics rather than Hamiltonian
operators themselves.

\subsection{One-component case}

Any one-component operator can be reduced to $\ddx{3}$, see \cite{GP91, GP97,
  Doyle}.  Indeed, in this case system ({\ref{eq:3}) implies $g_{11,1}=c_{111}=0$.

\subsection{Two-component case}
\label{sec:two-component-case}

Here we provide both affine and projective classifications. The main results
are summarised below.

\begin{theorem}\label{sec:two-component-case-1}
  Modulo (complex) affine transformations, the metric of any two-component
  homogeneous third-order Hamiltonian operator can be reduced to one of the
  three canonical forms:
  \begin{displaymath}
    g^{(1)}=\left(
      \begin{array}{cc}
	(u^{2})^{2}+1 & -u^{1}u^{2} \\-u^{1}u^{2} & (u^{1})^{2}
      \end{array}
    \right),\quad
    g^{(2)}=\left(
      \begin{array}{cc}
	-2u^{2} & u^{1} \\
	u^{1} & 0
      \end{array}
    \right),
    \quad
    g^{(3)}=\left(
      \begin{array}{cc}
	1 & 0 \\
	0 & 1
      \end{array}
    \right).
  \end{displaymath}%
\end{theorem}
 The metric $g^{(1)}$ gives rise to the Hamiltonian operator
\begin{equation*}
 \ddx{}\left(
\begin{array}{cc}
\ddx{} & \displaystyle\ddx{} \frac{u^2}{u^1} \\
\displaystyle\frac{u^2}{u^1}\ddx{} & \displaystyle
\frac{(u^2)^2+1}{2(u^1)^2}\ddx{}+\ddx{}\frac{(u^2)^2+1}{2(u^1)^{2}}
\end{array}%
\right) \ddx{},
\end{equation*}
 the metric $g^{(2)}$ corresponds to the  Hamiltonian operator from Example 1 of Section \ref{sec:examples},
\begin{equation*}
 \ddx{}\left(
\begin{array}{cc}
0 & \displaystyle\ddx{}\frac{1}{u^1} \\
\displaystyle\frac{1}{u^1}\ddx{} & \displaystyle
\frac{u^2}{(u^1)^{2}}\ddx{}+\ddx{}\frac{u^2}{(u^1)^{2}}
\end{array}%
\right) \ddx{}.
\end{equation*}
Note that these third-order operators are compatible (form a Hamiltonian pair).
The singular varieties of the first two metrics are double lines:
$(u^1)^2=0$. Applying a projective transformation which sends this line to the
line at infinity, one can reduce the first two cases to constant form. This
leads to

\begin{theorem}\label{sec:two-component-case-2}
  Modulo projective transformations, any two-component homogeneous third-order
  Hamiltonian operator can be reduced to constant coefficient form.
\end{theorem}

\begin{proof}[Proof of Theorem \ref{sec:two-component-case-1}]

  Setting $ U=(u^1du^2-u^2du^1, \ du^1, \ du^2), $ one can represent a generic
  two-component Monge metric in the form $ g=UQU^t $ where $Q$ is a constant
  $3\times 3$ symmetric matrix. Setting
$$
Q=\left(
  \begin{array}{ccc}
    r & h_1 &-h_2 \\
    h_1 & f_{11}&f_{12}\\
    -h_2&f_{12}&f_{22}
  \end{array}
\right)
$$
one obtains, explicitly,
\begin{equation}
  \label{eq:35}
  g= r(u^1du^2-u^2du^1)^2 + 2(u^1du^2-u^2du^1)(h_1du^1 - h_2du^2)
  +f_{ij}du^idu^j.
\end{equation}
Equations (\ref{nonlin})  impose a single cubic constraint,
\begin{equation}\label{eq:27}
  \det Q=r(f_{11}f_{22}-f_{12}^{2})-h_{1}^{2}f_{22}-h_{2}^{2}f_{11}-2h_{1}h_{2}f_{12}=0.
\end{equation}
We have the following cases.
\begin{description}
\item[Case 1:] $r\neq 0$. By scaling the dependent variables one can set
  $r=1$. By shifting the dependent variables one can also assume that
  $h_1=h_2=0$, so that the constraint \eqref{eq:27} simplifies to
  $f_{11}f_{22}-f_{12}^2=0$. This means that the constant part of the metric
  (\ref{eq:35}), $f_{ij}du^idu^j$, is degenerate, and therefore can be reduced
  to $(du^1)^2$ by an affine transformation.  This results in the metric
  \begin{equation*}
    g^{(1)}=\left(
      \begin{array}{cc}
	(u^{2})^{2}+1 & -u^{1}u^{2} \\
	-u^{1}u^{2} & (u^{1})^{2}
      \end{array}
    \right).
  \end{equation*}

\item[Case 2:] $r=0$. Modulo affine transformations one can always assume
  $h_1=1, \ h_2=0$ (if both $h_1$ and $h_2$ vanish we have the constant
  coefficient case which can be reduced to $g^{(3)}$).  Then (\ref{eq:27})
  implies $f_{22}=0$, and by shifting the dependent variables one can eliminate
  $f_{11}$ and $f_{12}$. This results in
  \begin{equation*} {g}^{(2)}=\left(
      \begin{array}{cc}
	-2u^{2} & u^{1} \\
	u^{1} & 0%
      \end{array}%
    \right).
  \end{equation*}%
\end{description}


It remains to point out that the metrics ${g}^{(1)}, {g}^{(2)}, {g}^{(3)}$ are
not affinely equivalent: the degrees of their coefficients are different.
\end{proof}


\begin{proof}[Proof of Theorem \ref{sec:two-component-case-2}]

  One can verify that the projective transformation
$$
\tilde u^1=\frac{1}{u^1}, ~~~ \tilde u^2=\frac{u^2}{u^1}, ~~~ \tilde
g^{}=\frac{g^{}}{(u^1)^4}
$$
reduces the first and the second metrics to constant coefficient forms
$(d{\tilde u}^1)^2+(d{\tilde u}^2)^2$ and $-2d{\tilde u}^1d{\tilde u}^2$,
respectively.  Thus, the corresponding third-order Hamiltonian operators can be
transformed to constant forms by one and the same reciprocal transformation. This  explains their compatibility.
\end{proof}

\subsection{Three-component case}
\label{sec:three-component-case}

Here we discuss the main result of this paper -  projective classification
of three-component Hamiltonian operators (affine classification would contain
too many cases and moduli).  This is achieved by going through the list of
normal forms of quadratic line complexes in $\mathbb{P}^3$, which fall into
eleven Segre types \cite{Jess}, and calculating the nonlinear constraints (\ref{nonlin}). Let us briefly recall the main setup. Consider the Pl\"ucker
quadric $ p^{12}p^{34} + p^{31}p^{24} + p^{14}p^{23}=0$, let $\Omega$ be the
$6\times 6$ symmetric matrix of this quadratic form. A quadratic line complex
is the intersection of the Pl\"ucker quadric with another homogeneous quadratic
equation in the Pl\"ucker coordinates, defined by a $6\times 6$ symmetric
matrix $Q$.  The key invariant of a quadratic complex is the Jordan normal form
of the matrix $Q\Omega^{-1}$. It is labelled by the Segre symbol which carries
information about the number and sizes of Jordan blocks.  Thus, the symbol
$[111111]$ indicates that the Jordan form of $Q\Omega^{-1}$ is diagonal; the
symbol $[222]$ indicates that the Jordan form of $Q\Omega^{-1}$ consists of
three $2\times 2$ Jordan blocks, etc.
We will also use `refined' Segre symbols with additional round brackets
indicating coincidences among the eigenvalues of some of the Jordan blocks,
e.g., $[(11)(11)(11)]$ denotes the subcase of $[111111]$ with three pairs of
coinciding eigenvalues, the symbol $[(111)(111)]$ denotes the subcase with two
triples of coinciding eigenvalues, etc.

The Monge metric results from the equation of the complex upon setting
$p^{ij}=u^idu^j-u^jdu^i$, and using the affine chart $ u^4=1, \ du^4=0$ (in
some cases it will be more convenient to use different affine charts, say,
$u^1=1, \ du^1=0$: this will be indicated explicitly where appropriate).  The
singular surface of a generic quadratic line complex in $\mathbb{P}^3$ is
Kummer's quartic surface, which can be defined as the degeneracy locus of the
corresponding Monge metric. For Monge metrics associated to third-order
Hamiltonian operators this quartic always degenerates into a double quadric
(which may further split into a pair of planes).

\begin{theorem}\label{sec:three-component-case-1}
  Modulo (complex) projective transformations, the Monge metric of any
  three-component homogeneous third-order Hamiltonian operator can be reduced
  to one of the six canonical forms:
  \begin{enumerate}
  \item Segre type $[(111)111]$: we have a one-parameter family of metrics
    ($c\ne \pm1$):
    \begin{equation*}
      g^{(1)}=\begin{pmatrix}
	(u^{2})^{2}+c & -u^{1}u^{2}-u^{3} & 2u^{2}
	\\ -u^{1}u^{2}-u^{3} & (u^{1})^{2}+c(u^{3})^{2} & -cu^{2}u^{3}-u^{1}
	\\ 2u^{2} & -cu^{2}u^{3}-u^{1} & c(u^{2})^{2}+1
      \end{pmatrix},
    \end{equation*}
    $\det g^{(1)} =(c+1)(c-1) (u^1u^2 - u^3)^2$, the singular surface is a
    double quadric.
  \item Segre type $[(111)12]$:
    \begin{equation*}
      g^{(2)} = \begin{pmatrix}
	(u^{2})^{2}+1 & -u^{1}u^{2}-u^{3} & 2u^{2} \\
	-u^{1}u^{2}-u^{3} & (u^{1})^{2} & -u^{1} \\
	2u^{2} & -u^{1} & 1
      \end{pmatrix},
    \end{equation*}
    $\det g^{(2)} = (u^1u^2 - u^3)^2 $, the singular surface is a double
    quadric.
  \item Segre type $[11(112)]$:
    \begin{equation*}
      g^{(3)} = \begin{pmatrix}
	(u^{2})^{2}+1 &  -u^{1}u^{2}&0 \\
	-u^1u^2 & (u^1)^2 & 0 \\
	0 & 0 & 1%
      \end{pmatrix},
    \end{equation*}
    $\det g^{(3)} = (u^1)^2$, the singular surface is a pair of double planes
    (one of them at infinity).
  \item Segre type $[(114)]$:
    \begin{equation}\label{eq:32}
      g^{(4)}=   \begin{pmatrix}
	-2u^2  & u^1  &  0
	\\
	u^1  & 0 &   0
	\\
	0 & 0  & 1
      \end{pmatrix}
    \end{equation}
    $\det g^{(4)} = - (u^1)^2$, the singular surface is a pair of double planes
    (one of them at infinity).
  \item Segre type $[(123)]$:
    \begin{equation*}
      g^{(5)}=\begin{pmatrix}
	-2u^2 & u^1 & 1
	\\
	u^1 & 1 &0
	\\
	1 & 0 & 0
      \end{pmatrix},
    \end{equation*}
    $\det g^{(5)} = -1 $, the singular surface is a quadruple plane at
    infinity.
  \item Segre type $[(222)]$:
    \begin{equation*}
      g^{(6)} =
      \begin{pmatrix}
	1 & 0 & 0\\ 0 & 1 & 0\\ 0 & 0 & 1
      \end{pmatrix},
    \end{equation*}
    $\det g^{(6)} = 1 $, the singular surface is a quadruple plane at infinity.
  \end{enumerate}
\end{theorem}


\begin{proof}[Proof of Theorem \ref{sec:three-component-case-1}]

  The classification is achieved by going through the list of Segre types of
  quadratic complexes and selecting those whose Monge metrics fulfil
  (\ref{nonlin}). In what follows we use the notation of \cite{fmoss}; more
  details on the projective classification of quadratic complexes in
  $\mathbb{P}^3$ can be found in \cite{Jess}.

  \medskip

  \textbf{Segre type $[111111]$.}  In this case the equation of the complex is
  \begin{multline*}
    \lambda_1(p^{12}+p^{34})^2- \lambda_2(p^{12}-p^{34})^2+ \lambda_3
    (p^{13}+p^{42})^2 - \lambda_4(p^{13} - p^{42})^2
    \\
    + \lambda_5(p^{14}+p^{23})^2- \lambda_6(p^{14}-p^{23})^2=0,
  \end{multline*}
  here $\lambda_i$ are the eigenvalues of $Q\Omega^{-1}$. The corresponding
  Monge metric is
  \begin{multline*} [a_1+a_2(u^3)^2+a_3(u^2)^2]
    (du^1)^2+[a_2+a_1(u^3)^2+a_3(u^1)^2](du^2)^2
    \\
    +[a_3+a_1(u^2)^2+a_2(u^1)^2] (du^3)^2+ 2[\alpha u^3-a_3u^1u^2] du^1du^2
    \\
    +2[\beta u^2-a_2u^1u^3] du^1du^3+2[\gamma u^1-a_1u^2u^3] du^2du^3,
  \end{multline*}
  where $a_1=\lambda_5-\lambda_6, \ a_2=\lambda_3-\lambda_4, \
  a_3=\lambda_1-\lambda_2, \ \alpha=\lambda_5+\lambda_6-\lambda_3-\lambda_4, \
  \beta=\lambda_1+\lambda_2-\lambda_5-\lambda_6 , \
  \gamma=\lambda_3+\lambda_4-\lambda_1-\lambda_2$, notice that
  $\alpha+\beta+\gamma=0$.  A direct computation shows that the only metrics
  which satisfy (\ref{nonlin}) are those for which the eigenvalues fulfil the
  relation (up to permutations of the $\lambda_i$):
  \begin{displaymath}
    \lambda_2=\lambda_3=\lambda_4=\frac{\lambda_1+\lambda_5+\lambda_6}{3}.
  \end{displaymath}
  Complexes of this type are denoted $[(111)111]$. Without any loss of
  generality one can set $\lambda_2=\lambda_3=\lambda_4=0, \ \lambda_1=1, \
  \lambda_5=\frac{c-1}{2}, \ \lambda_6=-\frac{c+1}{2}$ where $c$ is a parameter
  (note that one can add a multiple of the Pl\"ucker quadric to the equation of
  the complex). This results in the metric
  \begin{equation*}
    {g}^{(1)}=\begin{pmatrix}
      (u^{2})^{2}+c & -u^{1}u^{2}-u^{3} & 2u^{2}
      \\ -u^{1}u^{2}-u^{3} & (u^{1})^{2}+c(u^{3})^{2} & -cu^{2}u^{3}-u^{1}
      \\ 2u^{2} & -cu^{2}u^{3}-u^{1} & c(u^{2})^{2}+1
    \end{pmatrix}.
  \end{equation*}
  Note that although different values of $c$ correspond to projectively
  non-equivalent complexes, the corresponding singular surface, $\det
  g^{(1)}=0$, which is the double quadric $(u^1u^2 - u^3)^2=0$, does not depend
  on $c$. Non-equivalent complexes with coinciding singular surfaces are called
  cosingular. It was shown in \cite{A} that varieties of cosingular line
  complexes are generically curves, with the only exception provided by
  complexes of Segre type $[(111)111]$, in which case the variety of cosingular
  complexes is two-dimensional.  This explains, in particular, why the
  cosingular complex of Segre type $[(111)(111)]$, known as `special', does not
  occur in our classification.

  Another, more symmetric choice of representative within the same class, can
  be obtained if one assumes
  \begin{displaymath}
    \lambda_2=\lambda_4=\lambda_6=\frac{\lambda_1+\lambda_3+\lambda_5}{3},
  \end{displaymath}
  where without any loss of generality one can set
  $\lambda_2=\lambda_4=\lambda_6=0, \ \lambda_1=p, \ \lambda_3=q, \
  \lambda_5=r$, which results in the metric
  \begin{equation*}
    \begin{pmatrix}
      r+p(u^{2})^{2}+q(u^3)^2 & (r-q)u^3-pu^{1}u^{2} & (p-r)u^2-qu^1u^3 \\
      (r-q)u^3-pu^{1}u^{2} & q+p(u^{1})^{2}+r(u^{3})^{2} & (q-p)u^1-ru^{2}u^{3}
      \\ (p-r)u^2-qu^1u^3 & (q-p)u^1-ru^{2}u^{3} & p+q(u^1)^2+r(u^{2})^{2}
    \end{pmatrix},
  \end{equation*}
  recall that $ p+q+r=0$. In this case
  \[
\det g^{(1)} = pqr ((u^1)^2 + (u^2)^2 + (u^3)^2 + 1)^2,
\]
 the corresponding singular surface is the double  (imaginary)  sphere.

  \medskip

  \textbf{Segre type $[11112]$}.  The equation of the complex is
  \begin{multline*}
    \lambda_1(p^{12}+p^{34})^2- \lambda_2(p^{12}-p^{34})^2+ \lambda_3
    (p^{13}+p^{42})^2- \lambda_4(p^{13} - p^{42})^2+ 4 \lambda_5 p^{14}p^{23}+ (p^{14})^2=0.
  \end{multline*}
  The corresponding Monge metric is
  \begin{multline*} [\lambda (u^2)^2+\mu(u^3)^2+1] (du^1)^2+[\lambda
    (u^1)^2+\mu] (du^2)^2+[\mu(u^1)^2+\lambda] (du^3)^2
    \\
  +  2[\alpha u^3-\lambda u^1u^2] du^1du^2+2[\beta u^2- \mu u^1u^3] du^1du^3+
    2\gamma u^1 du^2du^3,
  \end{multline*}
  where $\lambda=\lambda_1-\lambda_2$, $\mu=\lambda_3-\lambda_4$,
  $\alpha=-\lambda_3-\lambda_4+2\lambda_5$,
  $\beta=\lambda_1+\lambda_2-2\lambda_5$, $\gamma=-\alpha-\beta$.  A direct
  computation shows that there are two subcases which satisfy
  (\ref{nonlin}). Up to permutations of the eigenvalues $\lambda_i$, we have

  \noindent{\bf Subcase $[(111)12]$:}

  \begin{equation*}
    \lambda_2 = \lambda_3 = \lambda_4 = (1/3)(\lambda_1+2\lambda_5).
  \end{equation*}
  Without any loss of generality one can set $\lambda_1=1, \ \lambda_2=
  \lambda_3=\lambda_4=0, \ \lambda_5=-1/2$. This results in the metric
  \begin{equation*}
    g^{(2)} = \begin{pmatrix}
      (u^{2})^{2}+1 & -u^{1}u^{2}-u^{3} & 2u^{2} \\
      -u^{1}u^{2}-u^{3} & (u^{1})^{2} & -u^{1} \\
      2u^{2} & -u^{1} & 1
    \end{pmatrix}.
  \end{equation*}

  \noindent {\bf Subcase $[11(112)]$:}
  \begin{displaymath}
    \lambda_3 = \lambda_4 = \lambda_5 = (1/2)(\lambda_1+\lambda_2).
  \end{displaymath}
  Without any loss of generality one can set $\lambda_1=1/2, \ \lambda_2=-1/2,
  \ \lambda_3=\lambda_4=\lambda_5=0$. This results in the metric
  \begin{equation*}
    g^{(3)} = \begin{pmatrix}
      (u^{2})^{2}+1 & -u^{1}u^{2}&0 \\
      -u^1u^2 & (u^1)^2 & 0 \\
      0 & 0 & 1%
    \end{pmatrix}.
  \end{equation*}

  \medskip

  \textbf{Segre type $[114]$}.  The equation of the complex is
\begin{equation*}
    \lambda_1(p^{12}+p^{34})^2- \lambda_2(p^{12}-p^{34})^2+
    4\lambda_3(p^{14}p^{23}+p^{42}p^{13})+2p^{14}p^{42}+4(p^{13})^2=0.
\end{equation*}
  Setting $p^{ij}=u^idu^j-u^jdu^i$ and using the affine chart $u^1=1, \ du^1=0$
  we obtain the associated Monge metric,
  \begin{multline*}
    \lambda (du^2)^2+[\lambda (u^4)^2+4] (du^3)^2+[\lambda (u^3)^2-2u^2]
    (du^4)^2
    \\
   + 2\alpha u^4 du^2du^3+2[u^4-\alpha u^3] du^2du^4-2\lambda u^3u^4 du^3du^4,
  \end{multline*}
  where $\lambda=\lambda_1-\lambda_2$,
  $\alpha=2\lambda_3-\lambda_1-\lambda_2$. A direct computation shows that the
  only metrics which satisfy (\ref{nonlin}) are those for which
  $\lambda_1=\lambda_2=\lambda_3$. Complexes of this type are denoted
  $[(114)]$. Without any loss of generality one can set
  $\lambda_1=\lambda_2=\lambda_3=0$, which results in the metric
$$
4 (du^3)^2-2u^2(du^4)^2+2u^4du^2du^4=0.
$$
Setting $u^4\to u^1, \ u^3\to u^3/2$ we obtain
\begin{displaymath}
  g^{(4)}=   \begin{pmatrix}
    -2u^2  & u^1  &  0
    \\
    u^1  & 0  &   0
    \\
    0 & 0  & 1
  \end{pmatrix}.
\end{displaymath}

\medskip

\textbf{Segre type $[123]$.} The equation of the complex is
\begin{multline*}
  -\lambda_1(p^{12}-p^{34})^2+ 4\lambda_2 p^{13}p^{42}+ 4(p^{13})^2+  \lambda_3(4p^{14}p^{23}+(p^{12}+p^{34})^2)+2p^{14}(p^{12}+p^{34})=0.
\end{multline*}
Setting $p^{ij}=u^idu^j-u^jdu^i$ and using the affine chart $u^1=1, \ du^1=0$
we obtain the associated Monge metric,
\begin{multline*}
  \lambda (du^2)^2+[\lambda (u^4)^2+4] (du^3)^2+[\lambda (u^3)^2+2u^3]
  (du^4)^2
  \\
  +2\alpha u^4 du^2du^3+2[1-\lambda u^3] du^2du^4+2[\gamma u^2-\lambda
  u^3u^4-u^4] du^3du^4,
\end{multline*}
where $\lambda=\lambda_3-\lambda_1$, $\alpha=2\lambda_2-\lambda_1-\lambda_3, \
\gamma=\lambda-\alpha$.  The only metrics of this form which fulfil
(\ref{nonlin}) are those for which $\lambda_1=\lambda_2=\lambda_3$. Complexes
of this type are denoted $[(123)]$.  Without any loss of generality one can set
$\lambda_1=\lambda_2=\lambda_3=0$, which results in the metric
$$
4 (du^3)^2+2u^3(du^4)^2+2 du^2du^4-2u^4du^3du^4.
$$
Setting $u^4\to i\sqrt 2u^1, \ u^3\to u^2/2, \ u^2\to -i u^3/\sqrt 2$ we obtain
\begin{displaymath}
  g^{(5)}=\begin{pmatrix}
    -2u^2 & u^1 & 1
    \\
    u^1 & 1 & 0
    \\
    1 & 0 & 0
  \end{pmatrix}.
\end{displaymath}

\medskip

\textbf{Segre type [222].} Here we have two (projectively dual) subcases, with
the equations
\begin{equation*}
  2\lambda_1 p^{12}p^{34}+ 2\lambda_2 p^{13}p^{42}+ 2\lambda_3 p^{14}p^{23}+(p^{12})^2+(p^{13})^2+(p^{14})^2=0,
\end{equation*}
and
\begin{equation*}
    2\lambda_1 p^{12}p^{34}+2\lambda_2 p^{13}p^{42}+2\lambda_3 p^{14}p^{23}+(p^{23})^2+(p^{24})^2+(p^{34})^2=0,
\end{equation*}
respectively.  Setting $p^{ij}=u^idu^j-u^jdu^i$ and using the affine chart
$u^1=1, \ du^1=0$ we obtain the associated Monge metrics,
\begin{displaymath}
  (du^2)^2+ (du^3)^2+(du^4)^2+ 2\alpha u^4 du^2du^3+2\beta u^3 du^2du^4+2\gamma
  u^2 du^3du^4,
\end{displaymath}
and
\begin{multline*}
  ((u^3)^2+(u^4)^2)(du^2)^2+
  ((u^2)^2+(u^4)^2)(du^3)^2+((u^2)^2+(u^3)^2)(du^4)^2
  \\
  +2(\alpha u^4-u^2u^3) du^2du^3+2(\beta u^3-u^2u^4) du^2du^4+2(\gamma u^2
  -u^3u^4) du^3du^4,
\end{multline*}
where $\alpha=\lambda_2-\lambda_1, \ \beta =\lambda_1-\lambda_3, \ \gamma
=\lambda_3-\lambda_2$.  In both cases the condition (\ref{nonlin}) implies
$\lambda_1=\lambda_2=\lambda_3$ (such complexes are denoted $[(222)]$),
however, the second metric becomes degenerate. This is the constant case $
g^{(6)}$.

\textbf{Other Segre types,} namely $[1113]$, $[1122]$, $[15]$, $[24]$, $[33]$,
$[6]$, do not correspond to homogeneous third-order Hamiltonian
operators. Thus, the only allowed Segre types are those for which:

\noindent (a) The Jordan normal form of $Q\Omega^{-1}$ contains at least three
Jordan blocks (that is, there are at least three entries in square brackets).

\noindent (b) There are three distinct Jordan blocks with the same eigenvalue
$\lambda$ (that is, there is a round bracket with three entries). According to
\cite{Jess}, p. 200, this implies that the singular surface of the
corresponding quadratic complex is a double quadric (possibly, reducible).

\noindent (c) The average of the remaining three eigenvalues (outside round
brackets) equals $\lambda$ (without any loss of generality one can set
$\lambda=0$).

As demonstrated above, the only Segre types that satisfy all these conditions
are $[(111)111], [(111)12], [11(112)], [(114)], [(123)], [(222)]$.  The types
$[15], [24], [33], [6]$ do not satisfy condition (a); the types $[1113],
[1122]$ satisfying (a)-(c) lead to degenerate Monge metrics.

\end{proof}

\section{Concluding remarks}

The main result of this paper is a complete classification of 3-component
homogeneous third-order Hamiltonian operators of differential-geometric type,
which was obtained based on the link to Monge metrics and quadratic line
complexes. Modulo projective equivalence, we found six types
of such operators. This was done by going through the list of normal forms of
quadratic complexes in $\mathbb{P}^3$, labelled by their Segre types
\cite{Jess}.  We observed that the necessary condition for a quadratic line
complex to be associated with a Hamiltonian operator is the degeneration of its
singular surface, known as Kummer's quartic, into a double quadric (which
itself may split into a pair of planes).

\begin{itemize}

\item The main challenge is to extend our classification to the general
  $n$-component case $n>3$. First of all, one can generate new examples
  of homogeneous Hamiltonian operators with arbitrary number of components by
 taking direct sums of operators with fewer components.  Thus, it is natural to restrict to the classification of irreducible operators. Although the relation to
  Monge metrics and quadratic line complexes is still available, there exist no
  reasonable `normal forms' for quadratic complexes even in $\mathbb{P}^4$. All
  known examples suggest however that singular varieties of quadratic complexes
  in $\mathbb{P}^n$ corresponding to third-order $n$-component Hamiltonian
  operators are not arbitrary, and degenerate into a double hypersurface of
  degree $n-1$. We recall that the singular variety of a generic quadratic line
  complex in $\mathbb{P}^n$ is a hypersurface of degree $2n-2$. For $n=3, 4$
  these varieties are known as Kummer's quartics in $\mathbb{P}^3$
  \cite{Kummer} and Segre sextic hypersurfaces in $\mathbb{P}^4$ \cite{Segre},
  respectively.

\item It would be interesting to construct first order Hamiltonian operators
  compatible with the third order operators found in this paper, and to
  investigate the corresponding integrable hierarchies. For all Examples from
  Section \ref{sec:examples} this was done in \cite{FN, OM98, KN1, KN2}, see
  also \cite{Strachan} for some results in the case of constant third-order
  operators $\eta^{ij}\ddx{3}$.

  \item Some $n$-component third-order
 Hamiltonian operators (\ref{casimir}) possess, in addition to the $n$
local Casimirs $u^i$, another $n$ nonlocal Casimirs of the form
$s^i=\psi^i_j({\bf u}) \ddx{-1}u^j$. Changing from the flat coordinates $u^i$
to the nonlocal variables $s^i$ one obtains a first-order constant coefficient
operator $\eta^{ij}\ddx{}$, thus establishing the Darboux theorem. Although this procedure works for all
Hamiltonian operators from Examples 1-3 of Section \ref{sec:examples} \cite{FN, OM98, KN1, KN2}, it does not seem to be universally applicable.  For general third-order Hamiltonian operators, the  Darboux theorem  is yet to be established.

\end{itemize}

\section*{Acknowledgements}

We thank G. Potemin for his involvement at the early stage of this work, and numerous discussions.
We also thank B. Dubrovin, M Marvan, O. Mokhov and A. Nikitin for useful comments. This research was
largely influenced by our friend and colleague Yavuz Nutku (1943-2010) to whom
we dedicate this paper.  MVP and RFV acknowledge financial support from GNFM of
the Istituto Nazionale di Alta Matematica, the Istituto Nazionale di Fisica
Nucleare, and the Dipartimento di Matematica e Fisica ``E. De Giorgi'' of the
Universit\`a del Salento. RFV would like to thank A.C. Norman for his support
on REDUCE. MVP's work was also partially supported by the RF Government grant
\#2010-220-01-077, ag. \#11.G34.31.0005, by the grant of Presidium of RAS
\textquotedblleft Fundamental Problems of Nonlinear
Dynamics\textquotedblright\  and by the RFBR grant 11-01-00197.

\end{document}